\documentclass[11pt,a4paper,reqno]{amsart}

\usepackage{times}
\usepackage{mathrsfs,mathtools}
\usepackage{enumerate}
\usepackage{stmaryrd} 
\usepackage{amsmath,amsthm,amssymb,amsfonts,accents} 
\usepackage[usenames,dvipsnames]{xcolor}
\usepackage{subfiles}
\usepackage{dsfont}
\usepackage[mathscr]{euscript}
\usepackage{graphicx}
\usepackage[colorlinks=true, linkcolor=tub-red ,citecolor=my-blue]{hyperref}
\usepackage[T1]{fontenc}
\usepackage[english]{babel}
\usepackage{todonotes}
\usepackage{geometry}
\usepackage{xspace,cancel}
\usepackage[capitalize]{cleveref}
\crefname{subsection}{Subsection}{Subsections}
\usepackage[numbers,sort&compress,longnamesfirst]{natbib}
\usepackage{lipsum}  
\usepackage{extarrows}
\usepackage{relsize}
\usepackage{bbm}
\usepackage{blkarray} 
\usepackage{pgf,tikz}
\usepackage{mathrsfs}
\usepackage[capposition=bottom]{floatrow}
\usepackage{relsize}
\usepackage{mathtools}
\usepackage{fancyhdr}
\usepackage{enumitem}
\usepackage{arydshln}

\usepackage{tikz}
\usetikzlibrary{arrows,positioning,shadows,shapes} 
\tikzset{
    >=stealth',
    punkt/.style={
           rectangle,
           rounded corners,
           draw=black, very thick,
           text width=6.5em,
           minimum height=2em,
           text centered},
    pil/.style={
           <->,
           thick,
           shorten <=2pt,
           shorten >=2pt,}
}

\allowdisplaybreaks

\def \R{\mathbb{R}}

\def\d{\mathrm{d}}

\definecolor{my-blue}{RGB}{0,85,150}
\definecolor{tub-red}{RGB}{187,14,31}

\makeatletter \@addtoreset{equation}{section}

\newtheorem{theorem}{Theorem}[section]
\newtheorem*{assumption*}{\assumptionName}
    \providecommand{\assumptionName}{}
    

\newtheorem*{condition*}{\conditionName}
    \providecommand{\conditionName}{}
    

\newtheorem{lemma}[theorem]{Lemma}
\newtheorem{proposition}[theorem]{Proposition}
\newtheorem{definition}[theorem]{Definition}

\theoremstyle{definition}

\numberwithin{equation}{section}
\numberwithin{figure}{section}
\numberwithin{table}{section}

\newcommand{\tnorm}[1]{\left\vert\kern-0.25ex\left\vert\kern-0.25ex\left\vert #1 \right\vert\kern-0.25ex\right\vert\kern-0.25ex\right\vert}
\newcommand{\tnormbig}[1]{\bigl\vert\kern-0.25ex\bigl\vert\kern-0.25ex\bigl\vert #1 \bigr\vert\kern-0.25ex\bigr\vert\kern-0.25ex\bigr\vert}

\newcommand{\ud}{\ensuremath{\mathrm{d}}}

\newcommand{\dx}{\ud x}

\begin{document}

\title[Improved model free bounds for multi-asset options]{Improved model-free bounds for multi-asset options using \\option-implied information and deep learning}

\author[E. Dragazi]{Evangelia Dragazi}
\author[S. Liu]{Shuaiqiang Liu}
\author[A. Papapantoleon]{Antonis Papapantoleon}

\address{Department of Mathematics, SAMPS, NTUA, 15780 Athens, Greece}
\email{kdragazi@mail.ntua.gr}

\address{Delft Institute of Applied Mathematics, EEMCS, TU Delft, 2628 Delft, The Netherlands \& ING, Amsterdam, The Netherlands.}
\email{S.Liu-4@tudelft.nl}

\address{Delft Institute of Applied Mathematics, EEMCS, TU Delft, 2628 Delft, The Netherlands \& Department of Mathematics, SAMPS, National Technical University of Athens, 15780 Zografou, Greece \& Institute of Applied and Computational Mathematics, FORTH, 70013 Heraklion, Greece}
\email{a.papapantoleon@tudelft.nl}

\thanks{ED and AP gratefully acknowledge the financial support from the Hellenic Foundation for Research and Innovation Grant No. HFRI-FM17-2152.
    The views expressed in this paper are personal views of the author (SL) and do not necessarily reflect the views or policies of his current or past employers.}

\keywords{Model-free bounds, multi-asset options, dependence uncertainty, additional information, option-implied information, superhedging duality, penalization, deep learning.}  

\subjclass[2020]{91G20, 91G60, 68T07}

\date{}

\begin{abstract}
We consider the computation of model-free bounds for multi-asset options in a setting that combines dependence uncertainty with additional information on the dependence structure.
More specifically, we consider the setting where the marginal distributions are known and partial information, in the form of known prices for multi-asset options, is also available in the market.
We provide a fundamental theorem of asset pricing in this setting, as well as a superhedging duality that allows to transform the maximization problem over probability measures in a more tractable minimization problem over trading strategies.
The latter is solved using a penalization approach combined with a deep learning approximation using artificial neural networks.
The numerical method is fast and the computational time scales linearly with respect to the number of traded assets.
We finally examine the significance of various pieces of additional information. 
Empirical evidence suggests that ``relevant'' information, \textit{i.e.} prices of derivatives with the same payoff structure as the target payoff, are more useful that other information, and should be prioritized in view of the trade-off between accuracy and computational efficiency.
\end{abstract}

\maketitle



\section{Introduction}

The field of mathematical finance has been going through a change of paradigm in the last decade.
In the ``old'' paradigm, a specific model was selected and was considered as an accurate representation of reality in financial markets. 
Then, the task of financial mathematicians was to compute various quantities of interest, such as option prices, hedging strategies, risk measures, and so forth, in the selected model.
However, academics, practitioners and regulators have understood that no model, irrespective of its complexity, can offer an accurate enough representation of reality.
Therefore, in the ``new'' paradigm, there is no specific model selected, while the tasks for a financial mathematician remain to compute various quantities of interest, such as option prices, hedging strategies, risk measures, and so forth, in this new setting.
This description should be taken with a bit of salt and pepper, of course, since the ``old'' paradigm will still be used in the future, while the ``new'' paradigm dates back already to \citet{bertsimas2002on} and \citet{hobson1998robust}.

There are various ways in which to understand and implement the statement that there is ``no specific model selected''.
On the one hand, we can start from a specific model and take into account uncertainty in various degrees; for example, parameter uncertainty within the Black--Scholes or another model specification, model uncertainty where \textit{e.g.} all diffusion models with unspecified drift and diffusion coefficients are considered, or, more generally, a class of probability measures with certain properties can be considered.
This branch of the literature corresponds to robust methods and methods taking model uncertainty into account. 
On the other hand, we can start from the observable data in the financial market and try to infer bounds on the quantities of interest by utilizing the data and making only structural assumptions about the market, such as the absence of arbitrage opportunities, but without making any assumption about the probabilistic model governing the data.
This branch of the literature corresponds to the model-free and data-driven methods.
A graphical representation these two approaches appears in Figure \ref{fig:models}.
These approaches interact via information and uncertainty; if we add enough uncertainty in the model-specific world, we shall (in theory) approach the model-free setting. 
Conversely, if we add sufficient information in the model-free world we shall approach (in theory, again) the model-specific setting.

\begin{figure}
\centering
\begin{tikzpicture}[node distance=1.5cm, auto,]
 
\node[punkt] (market) {parameters};
\node[punkt, inner sep=5pt, below=0.5cm of market]
    (formidler) {models};
\node[punkt, inner sep=5pt, below=0.5cm of formidler]
    (measure) {measures};

 \node[above=of market] (dummy) {};
 \node[punkt, right=of dummy] (t) {Model-specific}
   edge[pil, bend left=45] (market.east) 
   edge[pil, bend left=45] (formidler.east) 
   edge[pil, bend left=45] (measure.east); 
 \node[punkt, left=of dummy] (g) {Model-free}
   edge[pil, bend right=45] (market.west)
   edge[pil, bend right=45] (formidler.west)
   edge[pil, bend right=45] (measure.west);

\begin{scope}[transform canvas={yshift=.25em}]
    \draw (g) edge[->,thick] node[above] {information} (t);
\end{scope}

\begin{scope}[transform canvas={yshift=-.25em}]
    \draw (g) edge[<-,thick] node[below] {uncertainty} (t);
\end{scope}

\end{tikzpicture}
\caption{Illustration of the interaction between model-free and model-specific settings.}
\label{fig:models}
\end{figure}
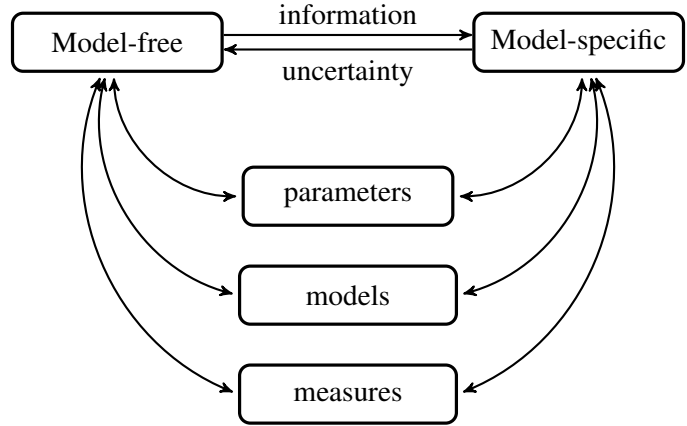

The literature on both approaches is vast and quickly expanding, hence we mention here only articles related to model-free bounds and model-uncertainty related to the valuation of multi-asset derivatives.
The readers interested in applications to quantitative risk management or the valuation of path-dependent derivatives (martingale optimal transport) should refer to \citet{Ruschendorf_Vanduffel_Bernard_2024} and \citet{HenryLabordere_2017} and the references therein. 

In the frameworks of dependence- and model-uncertainty and robust methods, several authors have derived bounds for multi-asset options using tools from probability theory, such as copulas and Fr\'echet--Hoeffding bounds, see \textit{e.g.} \citet{Chen_Deelstra_Dhaene_Vanmaele_2008, Dhaene_etal_2002_b, Dhaene_etal_2002_a}, and \citet{Hobson_Laurence_Wang_2005_2, Hobson_Laurence_Wang_2005_1}. 
These bounds were more recently enforced by additional information on the dependence structure, which led to the creation of improved Fr\'echet--Hoeffding bounds and the pricing of multi-asset options in the presence of additional information on the dependence structure, see \textit{e.g.}  \citet{tankov}, \citet{lux2016}, and \citet{puccetti2016}.
The setting of dependence uncertainty is closely linked with optimal transport theory, and its tools have also been used in order to derive bounds for multi-asset option prices, see \textit{e.g.} \citet{bartl2017marginal} for a formulation in the presence of additional information on the joint distribution. 
More recently, \citet{aquino2019bounds,eckstein2019computation}, and \citet{eckstein2019robust} have translated the model-free superhedging problem into an optimization problem over classes of functions, and used neural networks and the stochastic gradient descent algorithm for the computation of the bounds. 

In the framework of model-free and data-driven methods, ideas from operations research and optimization have been applied for the computation of model-free bounds.
\citet{bertsimas2002on} considered the computation of model-free bounds on a single-asset call option given the moments of the underlying asset price as well as the model-free bounds on a single-asset call option given other single-asset call and put option prices. 
\citet{dAspremont_ElGhaoui_2006} considered a framework where the prices of forwards and single-asset call options are known, and computed upper and lower bounds on basket options prices using linear programming. 
This work was later extended by various authors.
\citet{Pena_Vera_Zuluaga_2010} improved the results of \cite{dAspremont_ElGhaoui_2006} when computing the lower bounds on basket options prices in two special cases: (i) when the number of assets is limited to two and prices of basket options are known; and (ii) when the prices of only a forward and a single-asset call option per asset are known. 
\citet{pena2012computing} developed a linear programming-based approach for the problem of computing the upper price bound of a basket option given bid and ask prices of vanilla call options.
\citet{pena2010computing} studed the problem of computing the upper and lower bounds on basket and spread option prices when the prices of other basket and spread option prices are known. 
\citet{daum2011novel} developed a discretization-based algorithm for solving linear semi-infinite programming problems that returns a feasible solution, and applied the algorithm to compute the upper bounds on basket or spread options prices when single-asset call, put, and exotic options prices are known. 
\citet{Cho_Kim_Lee_2016} developed methods similar to \cite{daum2011novel} but for lower bounds on basket or spread options prices. 
Moreover,  \citet{Neufeld_Papapantoleon_Xiang_2023} considered a framework where bid and ask prices of single- and multi-asset options prices are observed, and developed efficient algorithms for the computation of model-free bounds that work in rather high-dimensions (\textit{e.g.} $d=60$).

The present article takes a hybrid approach between model-uncertainty and model-free and data-driven methods in finance.
We start by considering the classical setting of dependence uncertainty, where the marginals are considered (fully) known and the dependence structure is considered (fully) unknown.
In the classical setting of dependence uncertainty, bounds for multi-asset options have been derived by several methods, including Fr\'echet--Hoeffding bounds and copulas, and optimal transport theory.
However, these bounds are typically rather wide and not very informative. 
Moreover, there is usually additional information about the dependence structure in the market, in the form of traded multi-asset options, that is not taken into account.
We would thus like to extend the classical setting of dependence uncertainty by taking the additional, option-implied information into account, in order to have a realistic setting and deduce sharper bounds.
We take this information in a data-driven manner into account, that is, we assume there are some traded multi-asset options in the market, and then consider all the multivariate probability measures with known marginals that produce consistent prices for these options.  
The resulting framework offers a realistic and tractable approach to model uncertainty and data-driven methods when interested in the pricing of multi-asset derivatives.

Once the framework has been set-up, we provide a fundamental theorem of asset pricing in order to characterize the absence of arbitrage opportunities in this setting, and a superhedging duality.
The latter is an extension of the classical optimal transport duality, which is interesting in its own right, and also allows to translate an abstract maximization problem over probability measure into a more tractable minimization problem over trading strategies.
In order to solve numerically the latter minimization problem, we apply the method of \citet{eckstein2019computation}, that is, we first reduce the problem in a finite dimensional setting using neural networks and then enforce the inequality constraint by a penalization approach.

The numerical experiments conducted using artificial data show that the resulting numerical method is fast and accurate and, quite importantly, the computational time scales linearly with the number of traded assets.
There are also two interesting observations about the effect of additional information; in general, we can observe that more additional information makes the distance between the model-free bounds and the reference price tighter, which means that the bounds are more informative.
However, we can also observe in the numerical experiments, that ``relevant'' information, \textit{i.e.} traded options with the same payoff structure as the derivative we are interested in, result in a larger reduction of the distance between the model-free bounds and the reference price compared to other information, \textit{i.e.} traded options with different payoff structure.
This observation remains true even when some pieces of information are completely discarded and not used in the computation of the bounds.
Therefore, given the perennial trade-off between accuracy and efficiency, we can advise users to select and prioritize relevant information in the application of these methods. 

This article is organized as follows: in \cref{sec:FTAP-SHD}, we describe our framework that combines known marginals and additional information on the dependence structure, derive a fundamental theorem of asset pricing and deduce a superhedging duality.
In \cref{sec:optimal}, we briefly  characterize the optimal measures and trading strategies.
In \cref{sec:numerical-scheme}, we present the numerical method, which is based on a penalization approach and an approximation by neural networks.
Finally, in \cref{sec:results}, we present numerical experiments using artificial data and discuss their implications.


\section{Fundamental theorem and super-hedging duality}
\label{sec:FTAP-SHD}

Let us start by introducing some useful notions and notation.
Then, we will formulate and prove a fundamental theorem of asset pricing and a super-hedging duality in our setting, where the marginal distributions are known explicitly, while there is also additional information available in the financial market, in the form of traded prices of some multi-asset derivatives.

Consider a financial market where $d$ (primary) assets are traded and set $\mathcal{J}=\{1,\dots,d\}$.
Let $\mathscr{M}(\R^d)$ denote the set of all finite measures and $\mathscr{P}(\R^d)$ denote the set of all probability measures on the Borel $\sigma$-field of $\R^d$.
Let $\mathcal{I}$ be an arbitrary index set, $|\mathcal{I}|<\infty$, and  $\phi_i:\mathbb{R}^d \to \mathbb{R}$, $i \in \mathcal{I}$, be continuous, bounded and measurable functions.
We assume that these functions represent the payoff functions of traded multi-asset derivatives (\textit{i.e.} options depending on multiple primary assets) with traded prices $p_i\in\R_+$.
Therefore, any measure $\mu$ that is consistent with these option prices should satisfy
\[
    \int_{\mathbb{R}^d} \phi_i\d\mu =p_i, \quad\text{ for all } i\in\mathcal{I}. 
\]
In our setting, we assume that the marginal distributions $\nu_j$, $j\in\mathcal{J}$, of the asset prices (log-returns) are known explicitly, while there exists additional information in the market in the form of traded multi-asset derivatives with payoff $\phi_i$ and traded price $p_i$, $i\in\mathcal{I}$.
The following set contains all probability measures that are consistent with the market and option-implied information available in our setting
\begin{align}
\mathcal{Q}
    =\bigg\{ \mu\in\mathscr{P}(\mathbb{R}^d): \mu_j=\nu_j, j \in \mathcal{J}, \text{ and } \int_{\mathbb{R}^d} \phi_i \d\mu=p_i, i \in \mathcal{I} \bigg\}. 
\end{align}   

Let $f:\R^d\to\R$ denote a continuous, bounded and measurable function, which denotes the payoff function of a financial derivative. 
We would like to compute upper and lower model-free bounds for the price of this derivative in our setting, thus we would like to compute 
\begin{align*}
    \sup_{\mu\in\mathcal{Q}} \int_{\R^d} f \ud\mu
        \quad \text{ and } \quad 
    \inf_{\mu\in\mathcal{Q}} \int_{\R^d} f \ud\mu.
\end{align*}

On the `dual' side of the above optimization problem over probability measures, we have another optimization problem over trading strategies for the available assets and derivatives.
Let $\psi_j:\mathbb{R}\to\mathbb{R}$, $j \in \mathcal{J}$, be continuous, bounded and measurable functions, which denote the payoff functions of certain single-asset derivatives.
Let $b_i$ denote the amount of derivatives with payoff $\phi_i$ that is kept in the portfolio, hence $b=(b_1,\dots,b_{|\mathcal{I}|}) \in \mathbb{R}^{|\mathcal{I}|}$, and let $(\psi,b) = \big( (\psi_1,\dots,\psi_d),(b_1,\dots,b_{|\mathcal{I}|}) \big)$ denote the vector that summarizes these trading strategies. 
Moreover, we consider the set of trading strategies that dominate the target payoff $f$, \textit{\textit{i.e.}}
\begin{align}\label{eq:set-theta}
    \Theta(f) = \bigg\{ (\psi,b): \sum_{j\in\mathcal{J}} \psi_j + \sum_{i \in \mathcal{I}} b_i\phi_i \geq f \bigg\},
\end{align}
while the cost associated to such a trading strategy equals
\begin{align}\label{eq:cost}
    \pi(\psi,b) = \sum_{j\in\mathcal{J}} \int_{\mathbb{R}} \psi_j\d\nu_j + \sum_{i \in \mathcal{I}} b_ip_i.
\end{align}
Finally, let us introduce the following functional
\begin{align}\label{eq:primal}
    \Phi(f) = \inf \Big\{ \pi(\psi,b): (\psi,b) \in \Theta(f) \Big\},
\end{align}
which is the infimum over the cost of all trading strategies that dominate the target payoff $f$.

\begin{definition}
A trading strategy $(\psi, b)$ that satisfies 
\begin{equation*}
    (\psi,b) \in \Theta(\epsilon) \text{ \quad and \quad } \pi(\psi,b)=\sum_{j\in\mathcal{J}}\int_{\mathbb{R}}\psi_j\d\nu_j+\sum_{i \in \mathcal{I}} b_ip_i \leq 0
\end{equation*}
for some $\epsilon>0$, is called \emph{uniform strong arbitrage}.
\end{definition}

The main results of this section are a version of the first fundamental theorem of asset pricing and a superhedging duality, tailored to this model-free setting with additional information.

\begin{theorem}[Fundamental Theorem]\label{thm:FTAP}
There does not exist a uniform strong arbitrage strategy in the market if and only if the set $\mathcal{Q}$ is non empty.
\end{theorem}

\begin{theorem}[Superhedging Duality]\label{thm:SPH}
Let $f:\mathbb{R}^d \rightarrow \mathbb{R}$ be a continuous and bounded function. 
Assuming there does not exist uniform strong arbitrage in the market, then
\begin{equation*}
    \Phi(f)=\max_{\mu \in \mathcal{Q}}\int_{\mathbb{R}^d}f \d\mu.
\end{equation*}
\end{theorem}

The proof of these two theorems follows from a sequence of propositions and lemmata that are presented below.
A main tool for these proofs is the duality result of \citet{Bartl_Cheridito_Kupper_Tangpi_2017}, which states that every convex, increasing and continuous functional $\Psi:C_b \rightarrow \mathbb{R}$ admits the following representation 
\begin{equation*}
    \Psi(f) = \max \bigg\{ \int_{ \mathbb{R}^d} f \d\mu - \Psi^*(\mu),\ \ \mu \in \mathscr{M}(\mathbb{R}^d) \bigg\},
\end{equation*}
where $C_b$ denotes the set of continuous and bounded functions, $f:\mathbb{R}^d \rightarrow \mathbb{R}$, and $\Psi^*(\mu)=\sup \big \{ \int f \d\mu -\Psi(f) \big \}$ is the convex conjugate of $\Psi$. 

\begin{proposition}\label{dr}
The functional $\Phi(f)$ in \eqref{eq:primal} admits a dual representation, \emph{\textit{i.e.}}
\begin{equation*}
    \Phi(f) = \max\bigg\{ \int_{\mathbb{R}^d} f \d\mu - \Phi^*(\mu), \ \ \mu \in \mathscr{M} (\mathbb{R}^d) \bigg\},
\end{equation*}
for all $f:\mathbb{R}^d \rightarrow \mathbb{R}$ continuous and bounded.
\end{proposition}

\begin{proof}
Consider the functional $\Phi:C_b \rightarrow \mathbb{R}$ as defined in \eqref{eq:primal}, \textit{\textit{i.e.}}
\begin{equation*}
\begin{split}
\Phi(f)
    & = \inf \big\{ \pi(\psi,b): \ (\psi,b) \in \Theta(f) \big \}\\ 
    & = \inf \bigg\{ \sum_{j\in\mathcal{J}} \int_\R \psi_j \d \nu_j + \sum_{i \in \mathcal{I}} b_i\pi_i: \sum_{j\in\mathcal{J}} \psi_j + \sum_{i \in \mathcal{I}} b_i\phi_i \geq f \bigg\}.
\end{split}
\end{equation*}
Now, it suffices to show that $\Phi$ is a convex, increasing and continuous functional, and then we can employ \citet[Theorem 2.1]{Bartl_Cheridito_Kupper_Tangpi_2017} in order to conclude.

\textit{Step I.}
Initially, we need to show that the set $\Theta$ in \eqref{eq:set-theta} is homogeneous and additive.
In order to show the homogeneity of $\Theta$, \textit{\textit{i.e.}} that $\Theta(\lambda f) = \lambda \Theta(f)$ for $\lambda>0$, let $(\Tilde{\psi},\Tilde{b}) \in \lambda \Theta$. 
Then $(\Tilde{\psi},\Tilde{b})=(\lambda \psi, \lambda b)$ with $(\psi,b) \in \Theta(f)$. 
We want to show that $(\Tilde{\psi},\Tilde{b}) \in \Theta(\lambda f)$. 
Indeed we have that 
\begin{equation*}
\sum_j \Tilde{\psi}_j + \sum_i \Tilde{b}_i\phi_i
    = \sum_j \lambda \psi_j + \sum_i \lambda b_i \phi_i
    = \lambda \Big( \sum_j \psi_j + \sum_i b_i \phi_i \Big)
    \geq \lambda f,
\end{equation*}
which is true since $(\psi,b) \in \Theta(f)$.
Therefore $\lambda \Theta(f) \subset \Theta(\lambda f)$.
Conversely, let $(\Tilde{\psi},\Tilde{b}) \in \Theta(\lambda f), \textit{\textit{i.e.}} \sum_j\Tilde{\psi}_j+\sum_i\Tilde{b}_i\phi_i \geq \lambda f$. 
We want to show that $(\Tilde{\psi},\Tilde{b}) \in \lambda \Theta(f)$, in other words we want to show that there exists $(\psi,b) \in \Theta(f)$ such that $(\Tilde{\psi}, \Tilde{b})=(\lambda \psi,\lambda b)$, where $\sum_j\psi_j+\sum_ib_i\phi_i \geq f.$
Let us note that $\sum_j\Tilde{\psi}_j+\sum_i\Tilde{b}_i\phi_i \geq \lambda f$ is equivalent to $\sum_j\frac{1}{\lambda}\Tilde{\psi}_j+\sum_i\frac{1}{\lambda}\Tilde{b}_i\phi_i\geq f$, therefore there exists $(\psi=\frac{1}{\lambda}\Tilde{\psi}, b=\frac{1}{\lambda}\Tilde{b})$ such that $(\psi,b)\in \Theta(f)$. 
Hence $\Theta(\lambda f) \subset \lambda \Theta(f)$.

In order to show the additivity of $\Theta$, \textit{\textit{i.e.}} that $\Theta(f)+\Theta(g) = \Theta(f+g)$, let $(\psi,b) \in \Theta(f)+\Theta(g)$. 
Then, there exist $(\psi^1, b^1) \in \Theta(f)$ and $(\psi^2, b^2) \in \Theta(g)$ such that $\psi=\psi^1+\psi^2$ and $b=b^1+b^2$, in the sense that $\psi_j=\psi_j^1+\psi_j^2$, for all $j\in\mathcal{J}$, and $b_i=b_i^1+b_i^2$, for all $i\in\mathcal{I}$, where $\sum_j \psi_j^1+\sum_ib_i^1\phi_i  \geq f$ and $\sum_j\psi_j^2+\sum_i+b_i^2\phi_i\geq g$.
We want to show that $(\psi, b)\in \Theta(f+g)$, which is true since
\begin{equation*}
\begin{split}
\sum_j\psi_j + \sum_ib_i\phi_i 
    &= \sum_j(\psi_j^1+\psi_j^2) + \sum_i(b_i^1+b_i^2)\phi_i \\
    &= \Big (\sum_j \psi_j^1+\sum_i b_i^1 \phi_i\Big )+ \Big (\sum_j \psi_j^2+\sum_i b_i^2 \phi_i\Big ) 
    \geq f+g.
\end{split}
\end{equation*}
Conversely, let $(\psi,b) \in \Theta(f+g)$. 
Then we have $\sum_j \psi_j+\sum_i b_i\phi_i \geq f+g$.
We want to show that $(\psi, b) \in \Theta(f)+\Theta(g)$.
In other words, we want to show that there exist $(\psi^1,b^1) \in \Theta(f)$ and $(\psi^2,b^2) \in \Theta(g)$ such that $\psi=\psi^1+\psi^2$ and $b=b^1+b^2$.
However, since $\psi=\psi^+-\psi^-$ and $b_i=b_i^+-b_i^-$, where $\psi_j^+=\max\{\psi_j(x),0\}$, $\psi_j^-=\max\{-\psi_i(x),0\}$, for all $j\in\mathcal{J}$ and $b_i^+=\max\{b_i,0\}$, $b_i^-=\max\{-b_i,0\}$ for all $i \in \mathcal{I}$, from the hypothesis we get that
\begin{equation*}
f+g \le \sum_j\psi_j+\sum_ib_i\phi_i
    = \Big (\sum_j \psi_j^+ +\sum_i b_i^+\phi_i \Big ) + \Big( \sum_j(- \psi_j^-)+ \sum_i(-b_i^-\phi_i) \Big).
\end{equation*}
Therefore  $\Theta(f+g) \subset \Theta(f)+\Theta(g)$.

\textit{Step II.} We will show now that the functional $\Phi$ is linear and therefore convex. 
Let us first show that $\Phi(\lambda f)=\lambda \Phi(f)$.
Using the linearity of the cost function $\pi$ of a trading strategy in $\psi,b$, see \eqref{eq:cost}, we have immediately that
\begin{align*}
    \lambda \pi(\psi,b) = \pi(\lambda\psi,\lambda b).
\end{align*}
Hence, using also the homogeneity of the set $\Theta$ from the previous step, we get that
\begin{align*}
\Phi(\lambda f) 
    &= \inf \bigg \{ \pi(\psi,b) : (\psi,b) \in \Theta(\lambda f) \bigg \}\\
    &= \inf \bigg \{ \pi(\psi,b) : (\psi,b) \in \lambda \Theta( f) \bigg \}\\
    &= \inf \bigg \{ \pi(\lambda \Tilde{\psi},\lambda \Tilde b): (\Tilde{\psi},\Tilde{b}) \in  \Theta( f) \bigg \}\\
    &= \inf \bigg \{ \lambda \pi(\Tilde{\psi},\Tilde b): (\Tilde{\psi},\Tilde{b}) \in  \Theta( f) \bigg \}\\
    &= \lambda \inf \bigg \{ \pi(\Tilde{\psi},\Tilde b): (\Tilde{\psi},\Tilde{b}) \in  \Theta( f) \bigg \}
     = \lambda \Phi(f),
\end{align*}
where $\Tilde{\psi}=\frac{\psi}{\lambda}$ and $\Tilde{b_i}=\frac{b_i}{\lambda}$.

Next, we show that $\Phi(f+g) = \Phi(f) +\Phi(g)$. 
Using the additivity of the set $\Theta$, we have that
\begin{align*}
\Phi(f+g) 
    &= \inf \bigg\{ \pi(\psi,b): (\psi,b) \in \Theta(f+g) \bigg \}\\
    &= \inf \bigg\{ \pi(\psi,b): (\psi,b) \in \Theta(f)+\Theta(g) \bigg \}\\
    &= \inf \bigg\{\sum_{j\in\mathcal{J}} \int_{\mathbb{R}} (\psi_j^1 +\psi_j^2) \d \nu_j+\sum_{i \in \mathcal{I}}(b_i^1+b_i^2)p_i: (\psi^1,b^1) \in \Theta(f) , (\psi^2, b^2) \in \Theta(g) \bigg \}\\
    &= \inf \bigg\{ \pi(\psi^1,b^1): (\psi^1,b^1) \in \Theta(f)\bigg \} + \inf \bigg \{ \pi(\psi^2,b^2): (\psi^2,b^2) \in \Theta(g) \bigg \}
     = \Phi(f)+\Phi(g),         
\end{align*}
where $\psi_j^1=\max\{\psi_j,0\}=\psi_j^+$ and  $\psi_j^2=\max\{-\psi_j,0\}=\psi_j^-$.
                    
\textit{Step III.} We will show next that the functional $\Phi$ is increasing, \textit{i.e.} that for $f,g$ continuous and bounded functions with  $f \geq g$ holds $\Phi(f) \geq \Phi(g)$.
Let $(\psi,b) \in \Theta(f)$, then $\sum_j\psi_j+\sum_{i} b_i\phi_i \geq f \geq g$, hence $(\psi,b)\in \Theta(g)$, so that $\Theta(f) \subset \Theta(g)$.
Then, we have immediately that
\begin{align*}
\Phi(f) 
    &= \inf \bigg\{ \pi(\psi,b): (\psi,b)\in\Theta(f) \bigg \}\\
    &\geq \inf\bigg\{ \pi(\psi,b): (\psi,b)\in\Theta(g) \bigg \}
     = \Phi(g).
\end{align*}
  
\textit{Step IV.} Finally, we will show that the functional $\Phi$ is continuous, \textit{i.e.} that for all sequences $(f^n)_n\subset C_b$ with $f^n \searrow 0$ pointwise, then $\Phi(f^n) \searrow \Phi(0)=0$. 
Since $\Phi$ is linear, it suffices to show that it is bounded.
Let $f \in C_b$, without loss of generality we may assume that $\|f\|_{\infty} \leq 1$.
Then, for $\psi_1=\|f\|_{\infty}, \psi_2 = \dots = \psi_d = 0, b=0$, we have that $\Phi(f) \leq \|f\|_{\infty}\leq 1$ and $\Phi(-f)\leq \|-f\|_{\infty}=\|f\|_{\infty}\leq 1 $.
However, $\Phi$ is linear and $\Phi(-f)=-\Phi(f)$, thus $-\Phi(f)\leq 1$ yields that $\Phi(f)\geq 1$.
Therefore $\|\Phi\|_{\infty}=\sup_{f \in C_b, \|f\|_{\infty} \leq 1}|\Phi(f)|\leq 1$ and $\Phi$ is bounded.
\end{proof}

The proof of the superhedging duality is based on the following two results.

\begin{lemma}
Assume that the market is free of uniform strong arbitrage, then $\Phi(m)=m$, for all $m \in \mathbb{R}$.
\end{lemma}

\begin{proof}
We will show that the absence of uniform strong arbitrage yields that $\Phi(m)=m$, for all $m \in \mathbb{R}$, or, equivalently, that if $\Phi(m)<m-\epsilon$, for $\epsilon >0$, then there exists uniform strong arbitrage.
By definition, we have that $\Phi(m) \leq m$.
Indeed, since $\Phi(f)=\inf \big\{ \pi(\psi,b): (\psi,b)\in \Theta(f) \big\}$ is the infimum over the set $\Theta(f)= \big\{ (\psi,b): \sum_j\psi_j+\sum_{i} b_i \phi_i \geq f \big \}$, then for the trading strategy $(m,0,\dots,0,0,\dots,0)$ we have immediately that $(m,0,\dots,0,0,\dots,0) \in \Theta(m)$.
Now, consider an $\epsilon>0$ such that $\Phi(m)<m - \epsilon$.
Then there exists $(\psi,b) \in \Theta(m)$ such that 
\begin{align}\label{eq:phim}
    \sum_j \int_{\mathbb{R}} \psi_j\d\nu_j + \sum_{i} b_i \phi_i \leq m -\epsilon.
\end{align}
Let us define $\psi_j'(x):=\psi_j(x)-\frac{m-\epsilon}{d}$, for all $x \in \mathbb{R}$ and $j\in\mathcal{J}$, for the strategy $(\psi,b)\in\Theta(m)$. 
Then, we have that $(\psi',b) \in \Theta(\epsilon)$, since
\begin{align*}
\sum_j \psi_j' + \sum_{i} b_i\phi_i 
    &= \sum_j \Big\{ \psi_j-\frac{m-\epsilon}{d} \Big\} + \sum_i b_i\phi_i\\
    &= \sum_j \psi_j - d\frac{m-\epsilon}{d} + \sum_i b_i \phi_i\\
    &= \sum_j \psi_j + \sum_{i} b_i\phi_i -m+\epsilon
     \geq m-m+\epsilon = \epsilon.
\end{align*}
The cost of the strategy $(\psi',b)$ equals
\begin{align*}
\sum_j \int_{\mathbb{R}} \psi_j' \d\nu_j + \sum_{i} b_i p_i
    &= \sum_j \int_{\mathbb{R}} \Big\{ \psi_j-\frac{m-\epsilon}{d} \Big\} \d\nu_j + \sum_ib_ip_i \\
    &= \sum_j \int_{\mathbb{R}} \psi_j \d\nu_j + \sum_{i} b_i p_i - \sum_j \int_{\mathbb{R}} \frac{m-\epsilon}{d} \d\nu_j \\ 
    &= \sum_j \int_{\mathbb{R}} \psi_j \d\nu_j + \sum_{i} b_i p_i - d \frac{m-\epsilon}{d} \\
    &\overset{\eqref{eq:phim}}{\leq} m -\epsilon- (m-\epsilon)
     = 0.
\end{align*}
Therefore, we have constructed a trading strategy that is a uniform strong arbitrage, which is a contradiction.
Consequently, $\Phi(m)=m$ for all $m \in \mathbb{R}$.
\end{proof}

Proposition \ref{dr} shows that the functional $\Phi(f)$ admits the following representation:
\begin{equation*}
    \Phi(f) = \max\bigg\{ \int f \d\mu - \Phi^*(\mu), \ \mu \in \mathscr{M} (\mathbb{R}^d) \bigg\}.
\end{equation*}
Next, we want to show that the functional $\Phi^*$ equals zero for all measures $\mu$ in the set $\mathcal{Q}$, which settles the superhedging duality, \textit{i.e.} Theorem \ref{thm:SPH}.

\begin{proposition}\label{lem:SHD}
The conjugate functional $\Phi^*(\mu)$ of $\Phi(f)$ equals
\begin{equation*}
\Phi^*(\mu) := \sup_{f \in C_b} \bigg \{ \int_{\mathbb{R}^d} f \d\mu-\Phi(f) \bigg \}
    = \begin{cases}
        0, & \text{if } \mu \in \mathcal{Q}\\
        +\infty, & \text{otherwise}.
    \end{cases}
\end{equation*}
\end{proposition}

\begin{proof}
We will first show that for all finite measures $\mu$ in $\mathbb{R}^d$ that do not belong to the set $\mathcal{Q}$, the conjugate functional $\Phi^*$ is infinite. 

\textit{Case (i).} 
Assume first the $\mu$ is not a probability measure, \textit{i.e.} let $\mu(\mathbb{R}^d)>1$. 
Then, since $\Phi(m)=m$, we have immediately
\begin{equation*}
\Phi^* (\mu)
    \geq \sup_{m \in \mathbb{R}} \bigg\{ \int_{\mathbb{R}^d} m \d\mu-\Phi(m) \bigg\}
    = \sup_{m \in \mathbb{R}} \bigg\{ m \int_{\mathbb{R}^d}  \d\mu-m \bigg\}
    = \sup_{m \in \mathbb{R}} \bigg\{ m \mu(\mathbb{R}^d)-m \bigg\}
    = +\infty.
\end{equation*}
     
\textit{Case (ii).} 
Assume next that (some of) the marginals are not matching, \textit{i.e.} consider a set $B \in \mathcal{B}(\mathbb{R})$ and some $j \in \mathcal{J}$ such that $\mu_j(B) \neq \nu_j(B)$.
Then, there exists a function $h:\mathbb{R} \rightarrow \mathbb{R}$ bounded and continuous such that 
\[
    \int_{\R} h \d\mu_j > \int_{\mathbb{R}}h \d\nu_j.
\]
Let us also define the target payoff function $f(x):=h(x_j)$, $x\in \mathbb{R}^d$, which is continuous and bounded, and the trading strategy $(\psi,b)$ with
\[
    \psi_j(x) := h(x) \text{ and } \psi_k(x) := 0 \text{ for all } x\in\R, j,k\in\mathcal{J}, k \neq j, \text{ and } b_i=0 \text{ for all } i \in \mathcal{I}.
\]
Then, by construction, we have that $\sum_j \psi_j(x)+ \sum_i b_i\phi_i = \psi_j(x) = h(x) \geq f(x)$, hence $(\psi,b) \in \Theta(f)$, while
\[
    \pi(\psi,b) = \int_{\mathbb{R}}\psi_j \d\nu_j=\int_{\mathbb{R}}h \d\nu_j.
\]
Henceforth, we get that $\int f \d\mu - \Phi(f) \geq \int h\d\mu_j - \int h\d\nu_j>0$ and, using that $\Phi(\lambda f)=\lambda \Phi(f)$ for $\lambda>0$, we arrive at
\[
    \Phi^* (\mu) 
        \geq \sup_{\lambda >0} \bigg\{ \int_{\mathbb{R}^d} \lambda f \d\mu- \Phi(\lambda f) \bigg\}
        \geq \sup_{\lambda >0} \bigg\{ \lambda \bigg( \int_{\mathbb{R}} h \d\mu_j - \int_{\R} h \d\nu_j \bigg) \bigg\}
        = +\infty.
\]

\textit{Case (iii).} 
Assume now that (some of) the traded multi-asset prices are not matching, \textit{i.e.} let $\int_{\mathbb{R}^d}\phi_i\d\mu \geq p_i+ \epsilon$, for some $i \in \mathcal{I}$ and for $\epsilon>0$. 
Let $\delta\in \big [\frac{\epsilon}{1+d},\epsilon \big]$ and define the sets
\[
    A^i:=[A_1^i,+\infty) \times \dots \times [A_d^i,+\infty) 
       \ \text{ and } \
    C^i:=(A_1^i-\delta,+\infty) \times\dots \times (A_d^i-\delta,+\infty), 
\]
for the $i\in\mathcal{I}$ mentioned above, where $A_j^i\in\R$, for all $j\in\mathcal{J}$.
We assume, without loss of generality, that the multi-asset payoff function $\phi_i:\mathbb{R}^d \rightarrow \mathbb{R}$ satisfies the following condition
\[
\phi_i(x) = \begin{cases}
                \phi_i(x), & x \in A^i\\
                0,  & x \notin A^i.
            \end{cases}
\]
Let us also define the auxiliary function
\[
    g_i(x)=\phi_i(x)+\delta \mathbbm{1}_{(A_1^i-\delta, +\infty)}(x)+\dots +\delta \mathbbm{1}_{(A_d^i-\delta, +\infty)}(x),
\]
and observe that the following hold true:
\begin{align}\label{eq:help-26-iii}
\nu_j\big( (A_j^i-\delta,A_j^i) \big) < \frac{\epsilon - \delta}{d \delta} 
    \quad \text{ and } \quad
\nu_j\big( (A_j^i-\delta,+\infty) \big) < \frac{p_i - \int \phi_i \d\mu}{d \delta},
\end{align}
for $\delta \in \big [\frac{\epsilon}{1+d},\epsilon \big]$, $p_i-\int \phi_id \mu>\epsilon-\delta$ and $j\in\mathcal{J}$.
The sets $A^i$ and $(C^i)^c$ are closed and disjoint hence, by Urysohn's Lemma, there exists a function $f:\mathbb{R}^d \rightarrow \mathbb{R}$ such that $\phi_i \leq f \leq g_i$. 
Consider the strategy $(\psi,b)$ with $\psi_j(x):=\delta \mathbbm{1}_{(A_j^i-\delta,+\infty)}(x)$, for all $j\in\mathcal{J}$, and $b_i=1$ and $b_k=0$ for the aforementioned $i\in\mathcal{I}$ and all $k \in\mathcal{I}\setminus\{i\}$.
Then $\sum_j \psi_j + \sum_{i}b_i\phi_i \ge f(x)$ by construction, hence $(\psi,b) \in \Theta(f)$. 
Moreover, we have that
\begin{align*}
\Phi(f) 
    &\leq \sum_{j\in\mathcal{J}} \int_{\mathbb{R}}\psi_j\d\nu_j + \sum_{i\in\mathcal{I}} b_ip_i\\
    & = \sum_{j\in\mathcal{J}} \int_{\mathbb{R}} \delta \mathbbm{1}_{(A_j^i-\delta,+\infty)}\d\nu_j + p_i\\
    & = \delta \sum_{j\in\mathcal{J}} \bigg\{ \int_{\mathbb{R}} \mathbbm{1}_{(A_j^i-\delta,A_j^i)}\d\nu_j  
      + \int_{\mathbb{R}}\mathbbm{1}_{(A_j^i,+\infty)}\d\nu_j \bigg\} + p_i\\
    & = \delta \sum_{j\in\mathcal{J}} \bigg\{  \nu_j\big((A_j^i-\delta,A_j^i)\big)+\nu_j\big((A_j^i,+\infty)\big) \bigg\} + p_i\\
    & = \delta \sum_{j\in\mathcal{J}} \nu_j\big((A_j^i-\delta,A_j^i)\big)
      + \delta 
      - \delta \sum_{j\in\mathcal{J}} \nu_j\big((-\infty,A_j^i)\big) + p_i\\
    &\overset{\eqref{eq:help-26-iii}}{<} \delta d \frac{\epsilon-\delta}{d\delta}+\delta - \delta d \nu_{j_0}\big((-\infty,A_{j_0}^i)\big)+p_i\\
    & = \epsilon -\delta d \nu_{j_0}\big((-\infty, A_{j_0}^i)\big)+p_i,
\end{align*}
where $\nu_{j_0}\big((-\infty, A_{j_0}^i)\big) = \max_{j\in\mathcal{J}}\nu_j \big((-\infty, A_j^i) \big)$.
Henceforth, $-\Phi(f) \geq -p_i-\epsilon+ \delta \d\nu_{j_0}\big((-\infty, A_{j_0}^i)\big)$, which, together with the initial assumption that $\int_{\mathbb{R}^d}\phi_i\d\mu \geq p_i+ \epsilon$, yields immediately
\begin{equation*}
\int_{\mathbb{R}^d} f\d\mu - \Phi(f) 
    \geq  p_i + \epsilon - p_i - \epsilon + \delta d \nu_{j_0}\big( (-\infty, A_{j_0}^i)\big)
    = \delta d \nu_{j_0}\big( (-\infty, A_{j_0}^i)\big)
    >0.
\end{equation*}
Consequently, we have again that  
\[
    \Phi^*(\mu) \geq \sup_{\lambda >0} \bigg\{ \int_{\mathbb{R}^d} \lambda f \d\mu- \Phi(\lambda f) \bigg\} = +\infty.
\]
     
\textit{Case (iv).}
Finally, let us assume again that (some of) the prices of multi-asset derivatives are not matching, but now $\int_{\mathbb{R}^d}\phi_i\d\mu<p_i$ for some $i\in\mathcal{I}$.
The sets $A^i$ and $(C^i)^c$ defined previously are closed and disjoint, hence there exists a function $f:\mathbb{R}^d \rightarrow \mathbb{R}$ such that $-g_i \leq f \leq -\phi_i$. 
Then
\begin{equation*}
    -\int_{\mathbb{R}^d} g_i\d\mu \leq \int_{\mathbb{R}^d} f \d\mu \leq -\int_{\mathbb{R}^d}\phi_i\d\mu < p_i.
\end{equation*}
Using that $f\leq -\phi_i$, the strategy $(0,\dots,0,\dots,b_i=-1,\dots,0) \in \Theta(f)$,
and we have that $\Phi(f)\leq -p_i$ which is equivalent to $-\Phi(f) \geq p_i$. 
Henceforth, we have that
\begin{align*} 
\int_{\mathbb{R}^d}f \d\mu -\Phi(f) 
    &\geq - \int_{\mathbb{R}^d} g_i \d\mu + p_i\\
    &= -\int_{\mathbb{R}^d} \phi_i \d\mu -\delta \sum_{j\in\mathcal{J}}\nu_j \big( (A_j^i -\delta,+\infty) \big) + p_i\\
    & \geq p_i - \int_{\mathbb{R}^d} \phi_i \d\mu - \delta \d \nu_{j_1}\big( (A_{j_1}^i-\delta, +\infty) \big) 
    \overset{\eqref{eq:help-26-iii}}{>}0,
\end{align*}
where $\nu_{j_1}\big( (A_{j_1}^i-\delta,+\infty)\big)=\max_{j\in\mathcal{J}}\nu_j \big((A_j^i-\delta,+\infty) \big)$.
Consequently, we have once more that  
\[
    \Phi^* (\mu) \geq \sup_{\lambda >0} \bigg\{ \int_{\mathbb{R}^d} \lambda f \d\mu- \Phi(\lambda f) \bigg\} = +\infty.
\]

Finally, it remains to show that once $\mu \in \mathcal{Q}$ then $\Phi^*(\mu)=0$. 
Let us start by observing that for all probability measures $\mu$ we have that $\Phi^*(\mu) \geq 0$. 
Thus, we will just show that if $\mu \in \mathcal{Q}$ then $\Phi^*(\mu) \leq 0$.
Let $\mu \in \mathcal{Q}$, \textit{i.e.} $\mu(\mathbb{R}^d)=1$, $\mu_j=\nu_j$ for all $j\in\mathcal{J}$ and $\int_{\mathbb{R}^d}\phi_i \d\mu=p_i$, for all $i \in \mathcal{I}$, then, for all $(\psi,b)\in \Theta(f)$ we get that
\begin{align*}
\pi(\psi,b)
    &\ge \int_{\mathbb{R}^d} \bigg\{ \sum_{j \in \mathcal{J}} \psi_j(x_j) + \sum_{i \in \mathcal{I}} b_i\phi_i(x) \bigg\} \mu(\dx) 
     \geq \int_{\mathbb{R^d}}f(x)\mu(\dx).
\end{align*}
Consequently, we get directly that $\Phi(f) \geq \int f \d\mu$, therefore $\Phi^*(\mu)\leq 0$.
\end{proof}
       
Now, we can also provide the proof of the Fundamental Theorem.

\begin{proof}[Proof of Theorem \ref{thm:FTAP}]
We know from Proposition \eqref{dr} that $\Phi$ is linear, continuous and increasing, and admits the following dual representation
\begin{equation*}
    \Phi(f) = \max\bigg\{ \int f \d\mu - \Phi^*(\mu), \mu \in \mathscr{M} (\mathbb{R}^d) \bigg \}.
\end{equation*}
In case the financial market is free of uniform strong arbitrage, Proposition \ref{lem:SHD} yields that the dual representation reduces to 
\begin{equation*}
    \Phi(f)=\max_{\mu \in \mathcal{Q}}\int_{\mathbb{R}^d}f\d\mu.
\end{equation*}
Since $\Phi:C_b \rightarrow \mathbb{R}$ is a real function then this maximum exists, \textit{i.e.} the set $\mathcal{Q}$ is not empty.

Conversely, let $\mathcal{Q} \neq \emptyset.$ 
We will show that there does not exist a strategy which is a uniform strong arbitrage.
Indeed, let $\epsilon>0$ and consider a trading strategy $(\psi,b)$ such that $(\psi,b) \in \Theta(\epsilon)$. 
However, since $\mathcal{Q} \neq \emptyset$ we have that 
\begin{align*}
\pi(\psi,b)
    &= \sum_j \int_{\mathbb{R}}\psi_j \d\nu_j + \sum_{i} b_ip_i
     \geq \int_{\mathbb{R}^d} \bigg\{ \sum_j \psi_j(x_j) + \sum_i b_i\phi_i(x) \bigg\} \mu(\dx)\\
    &\geq \int_{\mathbb{R}^d} \epsilon \mu(\dx)
     = \epsilon.
\end{align*}
Hence, there does not exist uniform strong arbitrage.
\end{proof}


\section{Characterization of optimal measures and strategies}
\label{sec:optimal}

Let us now provide an interesting characterization of the optimal measures $\mu^\star$ and the optimal trading strategies $(\psi^\star,b^\star)$, \textit{i.e.} those that attain the supremum and the infimum in the superhedging duality (Theorem \ref{thm:SPH}).
A result analogous to \citet[Proposition 2.8.]{bartl2017marginal} holds in our setting, since the optimality conditions are satisfied, \textit{i.e.} the payoff function $f$ is bounded, the set $\mathcal{I}$ is finite, the set $\mathcal{Q}$ is non-empty, there exists a coupling $\mu$ between the marginals $\mu_1, \dots, \mu_d$ and $\int_{\mathbb{R}^d} \phi_i \d \mu=p_i$; the proof is omitted for the sake of brevity. 
Hence, there exists a strategy $(\psi^\star,b^\star) \in \Theta(f)$ such that the minimum price of the portfolio is attained, in other words we have that $\Phi(f)=\pi(\psi^\star,b^\star)$.
The following result shows that we can also characterize the optimal measure $\mu^{\star} \in \mathcal{Q}$ in terms of the trading strategies. 

\begin{proposition}
Let $f: \mathbb{R}^d \rightarrow \mathbb{R}$ be bounded, assume that the set $\mathcal{Q} \neq \emptyset$, and that for every $b \neq 0$ there exists a coupling $\mu$ between $\mu_1,\dots,\mu_d$. 
Then, for a probability measure $\mu^{\star} \in \mathcal{Q}$ the following are equivalent:
\begin{enumerate}
\item There exists a strategy $(\psi,b)\in \Theta(f)$ such that
    \begin{equation*}
        \sum_{j \in \mathcal{J}} \psi_j(x_j) + \sum_{i \in \mathcal{I}} b_i\phi_i(x) = f(x) \quad \text{ for } \mu^{\star}-\text{almost all } x \in \mathbb{R}^d.
    \end{equation*}
\item The probability measure $\mu^{\star}$ is optimal, that is, $\Phi(f)=\int f\d\mu^{\star}$. 
\end{enumerate}
\end{proposition}

\begin{proof}
Assume that $\mu^{\star}\in\mathcal{Q}$ is optimal and let $(\psi,b) \in \Theta(f)$ be the optimal trading strategy, \textit{i.e.} $\Phi(f)=\pi(\psi, b)$. 
Then, by definition, $\sum_j \psi_j + \sum_{i} b_i\phi_i \geq f$.
In order to show the reverse inequality, notice that
\begin{equation*}
\int_{\mathbb{R}^d} f \d\mu^{\star}
    = \Phi(f)
    = \pi(\psi,b)
    = \int_{\mathbb{R}^d} \bigg\{ \sum_{j \in \mathcal{J}} \psi_j + \sum_{i \in \mathcal{I}} b_i\phi_i \bigg\} \mu^{\star}(\dx)
\end{equation*}
Hence, for $\mu^{\star}-$ almost all $x \in \mathbb{R}^d$, it should hold $f \geq \sum_j \psi_j + \sum_{i} b_i\phi_i$.

Conversely, let $(\psi,b) \in \Theta(f)$ such that  $\sum_j \psi_j(x_j) + \sum_{i} b_i\phi_i(x) = f(x)$ for $\mu^{\star}-$ almost every $x \in \mathbb{R}^d$. 
Using that the infimum is attained, we will show that $\pi(\psi,b)=\int f \d \mu^{\star}$. 
Indeed, we have that
\begin{equation*}
\int_{\mathbb{R}^d}f \d\mu^{\star}
    \leq \sum_{j \in \mathcal{J}} \int_{\mathbb{R}} \psi_j\d\nu_j + \sum_{i \in \mathcal{I}} b_ip_i
    = \int_{\mathbb{R}^d} \bigg\{ \sum_{j \in \mathcal{J}} \psi_j + \sum_{i \in \mathcal{I}} b_i\phi_i \bigg\} \mu^{\star}(\dx)
    = \int_{\mathbb{R}}f \d\mu^{\star}. \qedhere
\end{equation*}
\end{proof}


\section{Computations using penalization and deep learning}
\label{sec:numerical-scheme}

Let us now explain the numerical scheme we will use in order to compute model-free bounds for option prices in our setting.
Our method is based on the penalization and deep learning scheme developed by \citet{eckstein2019computation}.
In order to compute, for example, the upper model-free bound of a derivative with payoff $f$ in our setting, \textit{i.e.}
\[
    \sup_{\mu\in\mathcal{Q}} \int_{\R^d} f\ud \mu,
\]
which is an optimization problem over probability measures, that is usually difficult to discretize and solve numerically, we will turn instead to the `dual' problem, \textit{i.e.}
\begin{align*}
\Phi(f) &= \inf \Big\{ \pi(\psi,b): (\psi,b) \in \Theta(f) \Big\} \\
        &= \inf \bigg\{ \sum_{j\in\mathcal{J}} \int_{\mathbb{R}} \psi_j\d\nu_j + \sum_{i \in \mathcal{I}} b_ip_i : \sum_{j\in\mathcal{J}} \psi_j + \sum_{i \in \mathcal{I}} b_i\phi_i \geq f\bigg\}.
\end{align*}
This is a simpler problem to solve numerically, since it involves the optimization over continuous and bounded functions ($\psi_j$) and vectors ($b_i$).

The first step is to reduce this problem to a finite dimensional setting, by replacing the space $\mathcal{H}=\{h: C_b \rightarrow \mathbb{R}, \text{ continuous and bounded}\}$ where the functions $\psi_j, j\in\mathcal J$, take place, by a subset $\mathcal{H}^m$.  
The space $\mathcal{H}^m$ can be described by a set of neural networks with a fixed structure, but unspecified parameter values, where $m$ counts the number of neurons per layer.
The sequence $(\mathcal{H}^m)_{m \in \mathbb{N}}$ is increasing, $\mathcal{H}^{\infty}= \cup_{m \in \mathbb{N}} \mathcal{H}^m$, and $\mathcal{H}^\infty$ is dense in $\mathcal{H}$.
Then, the optimization problem takes the form
\begin{align*}
\Phi^m(f) 
    &= \inf \bigg\{ \sum_{j\in\mathcal{J}} \int_{\mathbb{R}} \psi_j\d\nu_j + \sum_{i \in \mathcal{I}} b_ip_i : \sum_{j\in\mathcal{J}} \psi_j + \sum_{i \in \mathcal{I}} b_i\phi_i \geq f, \psi_j\in\mathcal{H}^m, j\in\mathcal{J} \bigg\}.
\end{align*}

The second step, in order to allow a step-wise updating of the parameters, is to penalize the inequality constraint $\sum_{j} \psi_j + \sum_{i} b_i\phi_i \geq f$ using a non-decreasing penalty function $\beta: \mathbb{R} \rightarrow \mathbb{R}_+$ and considering a sequence of penalty functions $(\beta_{\gamma})_{\gamma >0}$ parameterized by a penalty factor $\gamma$.  
Let us introduce a reference probability measure $\theta$, then we get the penalized and parameterized optimization problem
\begin{align*}
\Phi^m_{\theta, \gamma}(f) 
    &= \inf \bigg\{ \sum_{j\in\mathcal{J}} \int_{\mathbb{R}} \psi_j\d\nu_j + \sum_{i \in \mathcal{I}} b_ip_i + \int_{\R} \beta_\gamma \bigg( f - \sum_{j\in\mathcal{J}} \psi_j - \sum_{i \in \mathcal{I}} b_i\phi_i \bigg) \d\theta, \psi_j\in\mathcal{H}^m, j\in\mathcal{J} \bigg\}.
\end{align*}
The optimization problem $\Phi^m_{\theta, \gamma}(f)$ for fixed $m$ and $\gamma$ is a tractable optimization problem involving a class of neural networks for approximating the continuous functions $\psi_j$ and, according to \citet{eckstein2019computation}, we have the following convergence result
\[
    \Phi^m_{\theta, \gamma}(f) \longrightarrow \Phi(f) \quad \text{ as } m, \gamma \rightarrow +\infty.
\]


\section{Numerical experiments and results}
\label{sec:results}

Let us now present the numerical experiments we have conducted using artificial data, and discuss the results.
We would like to evaluate how the presence of additional information allows to reduce the no-arbitrage gap, \textit{i.e.} the difference between the reference price and the model-free bound or the upper and lower model-free bounds, to understand whether some pieces of additional information are more useful than others, and finally to test how the method scales with the dimension of the problem (\textit{i.e.} the number of assets).

Throughout this section, we will train deep neural networks with the following characteristics in order to compute the model-free upper bound of multi-asset options:

\begin{table}[H]
\scalebox{1.0}{
\textbf{}\begin{tabular}{  c | c}
    Hyper-parameter     & Value \\ 
    \hline\hline
    Hidden layers       & 4  \\
    Neurons per layer   & 64 \\
    Activation function & ReLU\\       
    Initialization      & Xavier\\
    Optimizer           & Adam\\
    Batch size          & 128\\
    {Iterations } & 25,000\\
    \hdashline
    $\gamma$  & 80 \\
    $\theta$  & normal\\
    \hline	  
\end{tabular} }
\caption{Penalization and neural network hyper-parameters.}\label{tab:hyper}
\end{table}

\noindent In order to generate artificial data for our experiments, we will use the Black--Scholes model as the benchmark model, \textit{i.e.} we will assume that the marginals of the log-returns follow the normal distribution and are coupled with a Gaussian copula with correlation matrix $\rho$.
Then, we will generate the prices for the options that represent the additional information using this benchmark model and we will also compute the reference price using the same model. 
Subsequently, we will discard the benchmark model and compute the model-free bounds using only the option prices generated from this, exactly as postulated in our method.


\subsection{Experiment 1}
\label{subseq:ex-1}

Consider a financial market that contains three traded assets $S=(S^1, S^2, S^3)$ and assume we are interested in the price of a European call-on-max option on these assets, \textit{i.e.} the payoff function equals
\begin{equation} \label{eq:exp1}
    f(x) = \big( x_1 \vee x_2 \vee x_3 - K \big)^+.
\end{equation}
The interest rate is assumed zero for simplicity, the maturity $T=1.5$ (years), while the initial values, variances and correlation matrix are summarized below
\begin{align*}
S_0 = \begin{pmatrix}
    10 \\
    10 \\
    10
    \end{pmatrix}, 
\quad
\sigma = \begin{pmatrix}
    0.3 \\
    0.4 \\
    0.5
    \end{pmatrix}, 
\quad    
\rho = \begin{pmatrix}
    1 & 0.5 & 0.5 \\
    0.5 & 1 & 0.5 \\
    0.5 & 0.5 & 1 
    \end{pmatrix}.
\end{align*}

In order to understand the impact that additional information has on the model-free bounds, we will consider several cases where we gradually add more information, \textit{i.e.} more traded options with known prices.
More specifically, we will consider the following cases:
\begin{enumerate}[label=(E1.\arabic*),itemindent=1.em]\setcounter{enumi}{-1}
\item\label{E11} Base case: only the three marginal distributions are known.
\item Case 1: Base plus additional traded two-asset call-on-max options with payoff function 
 $$\phi_1 (x) = \big( x_1 \vee x_2 - K \big)^+, \quad K=6.$$ 
\item Case 2: Case 1 plus additional traded two-asset call-on-max options with payoff function
    $$\phi_2 (x) = \big( x_2 \vee x_3 - K \big)^+, \quad K=6.$$ 
\item Case 3: Case 2 plus additional traded two-asset call-on-max options with payoff function    
    $$\phi_3 (x) = \big( x_1 \vee x_3 - K \big)^+, \quad K=\{5,6,7\}.$$ 
\item\label{E15} Case 4: Case 3 plus additional traded three-asset call-on-max options with payoff function
     $$\phi_4 (x) = \big( x_1 \vee x_2 \vee x_3 - K \big)^+, \quad K=\{5,7\}.$$ 
\end{enumerate}


\subsubsection*{Convergence}

In this first experiment, we would like to perform a ``sanity check'' of the method, to ensure that the penalization and neural network approximation with the parameters selected converges to the correct value for the bounds.
Let us thus consider the numerical scheme with the hyper-parameters outlined in \cref{tab:hyper} and assume we are interested in computing the model-free upper bound for a call-on-max option with payoff \eqref{eq:exp1} and strike $K=6$.

\begin{figure}[H]
\centering
    \includegraphics[width=0.6\textwidth]{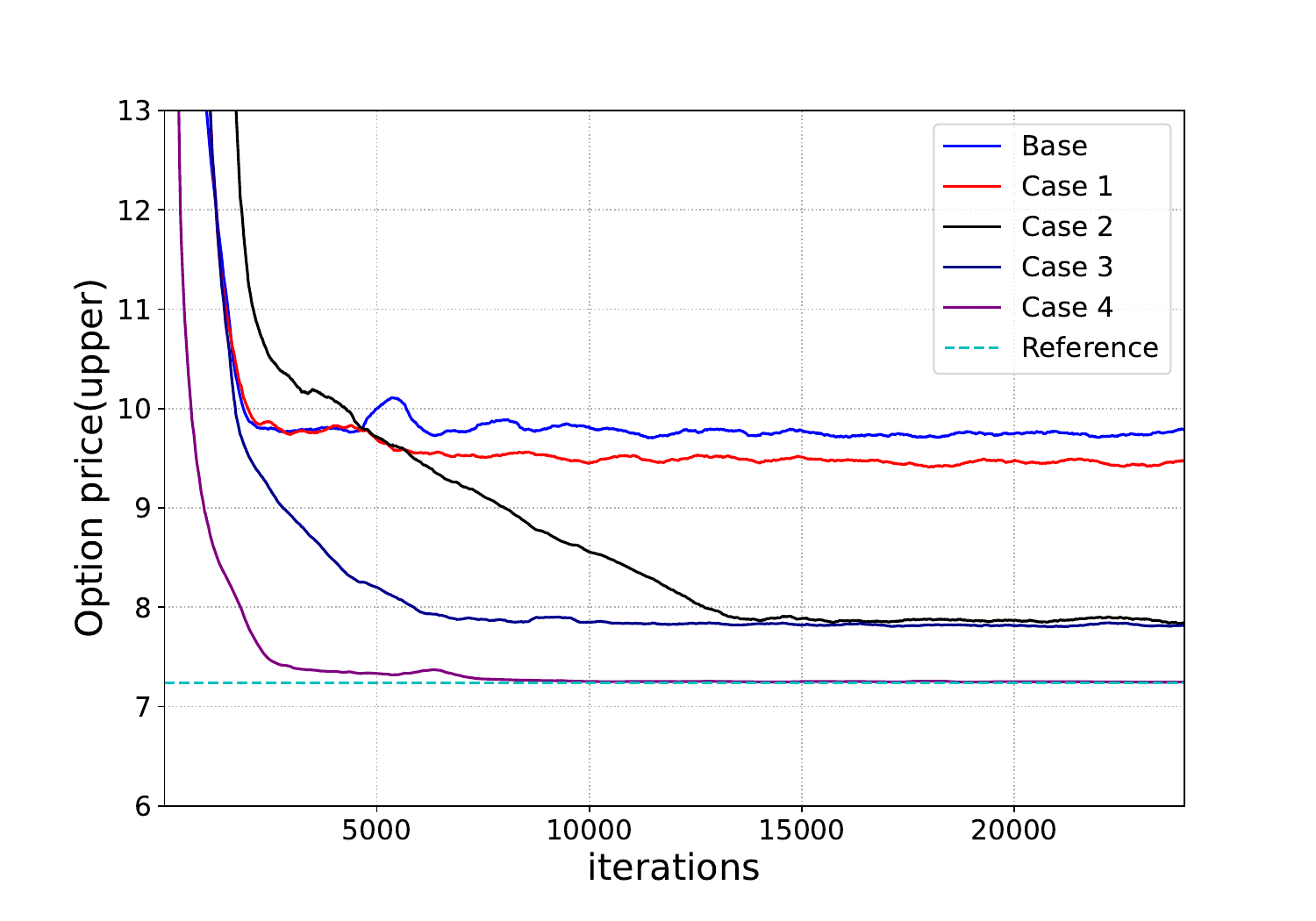}    
    \caption{Model-free upper bounds for the price of a three-asset call-on-max option with strike price $K=6$, in the different cases outlined above.}
    \label{fig:ex1_1}
\end{figure}

The outcome of this experiment is presented in Figure \ref{fig:ex1_1}.
The immediate observation is that the numerical scheme has converged to the solution after roughly 15 thousand iterations.
This observation has been also confirmed by additional experiments that have been conducted in this setting, with different values for the model parameters and the hyper-parameters. 
In the numerical experiments that follow, we will thus work with the setting outlined in \cref{tab:hyper}, and use 25 thousand iterations for training the neural network.


\subsubsection*{Analysis and discussion}

In the experiments that follow, we would like to explore the impact that additional information has on the quality of the model-free bounds.
We would like to observe, in particular, how the structure of the additional information impacts the distance between the model-free upper bound and the reference price.

The left panel on Figure \ref{fig:ex1_2a} presents the model-free bounds for a three-asset call-on-max option for various strikes, in the setting outlined in \ref{E11}--\ref{E15}.
The strike prices for the target payoff function range from 2 to 14, while the strike prices for the additional information are concentrated on values less than 8. 
We can observe thus two main patterns: (i) the additional information is improving the model-free bounds, especially for strike prices less than 8; (ii) the impact of the additional information is deteriorating for options with larger strikes prices, and for strikes above 13 approximately all cases collapse to the base case. 

\begin{figure}[H]
    \centering
    \includegraphics[width=0.495\textwidth]{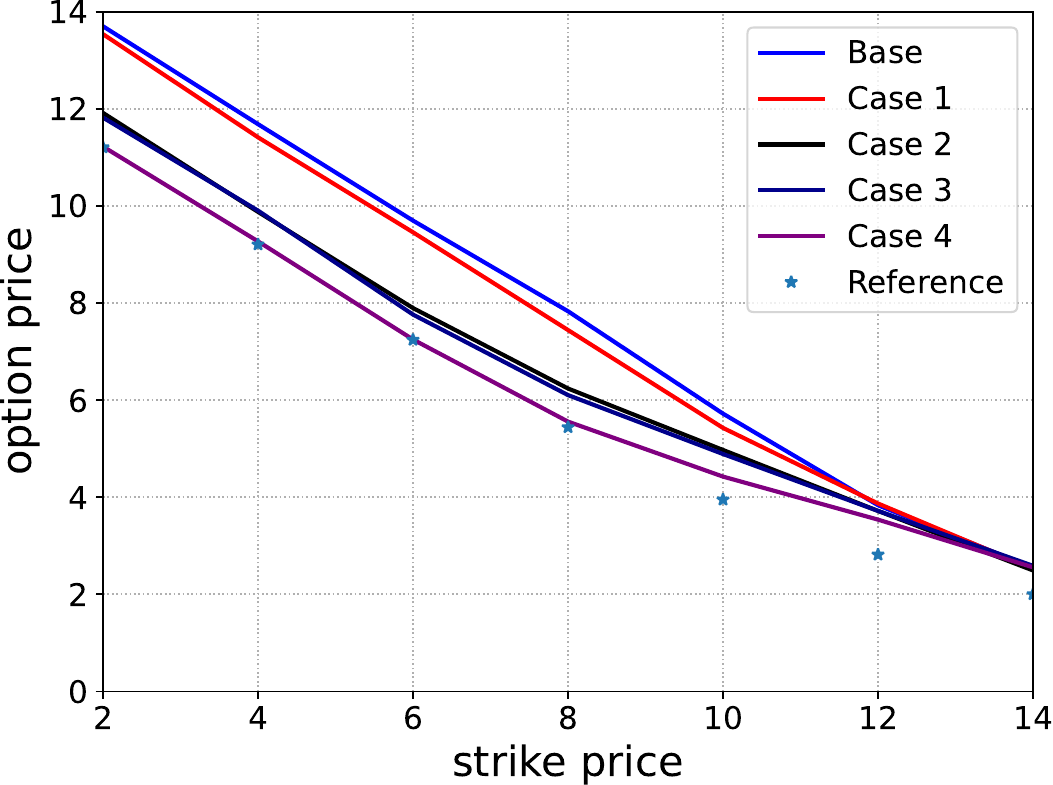}
    \includegraphics[width=0.495\textwidth]{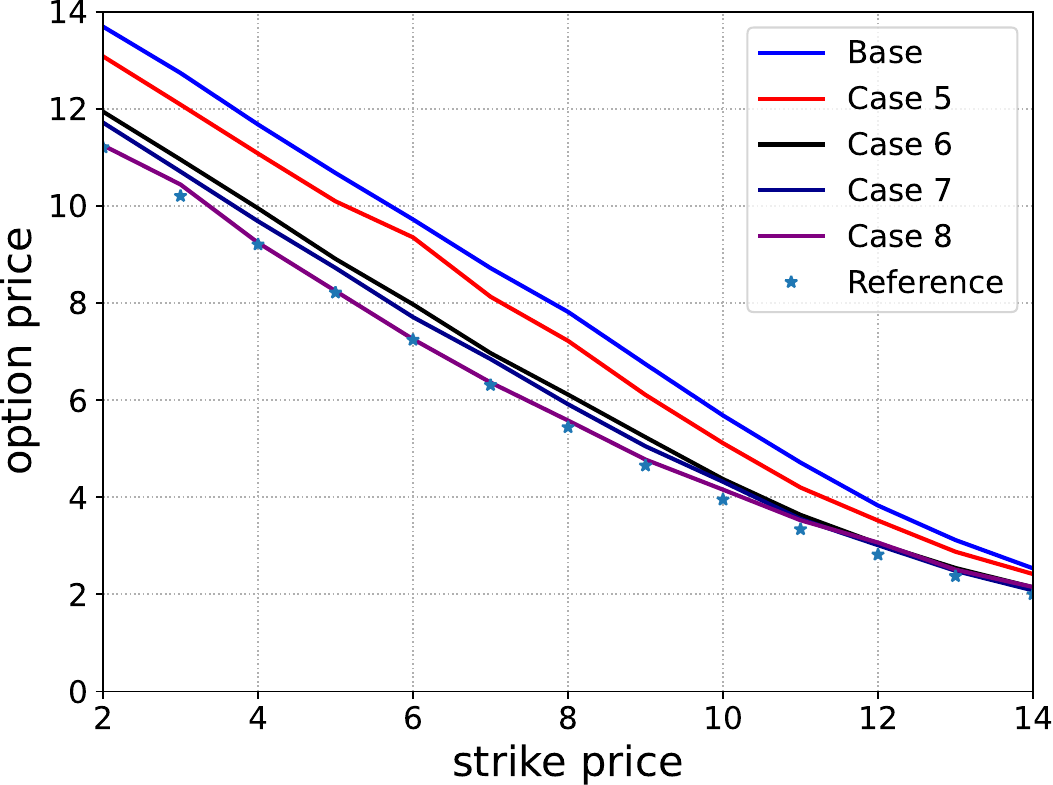}
    \caption{Model-free bounds for various strikes using the setting \ref{E11}--\ref{E15} on the left, and using the setting \ref{E16}--\ref{E19} on the right.}
    \label{fig:ex1_2a}
\end{figure}

The right panel on Figure \ref{fig:ex1_2a} presents the model-free bounds for a three asset call-on-max option for various strikes in the setting outlined in \ref{E16}--\ref{E19} below.
This setting supersedes \ref{E11}--\ref{E15} as follows: the base case remains the same and in every other case more options are added, namely 
\begin{enumerate}[label=(E1.\arabic*),itemindent=1.em]\setcounter{enumi}{4}
    \item\label{E16} Case 5: Base case plus 5 more options with payoff $\phi_1$ and strike prices $K=\{6, 9, 11, 13, 15\}$.
    \item Case 6: Case 5 plus 4 more options with payoff $\phi_2$ and strike prices $K=\{6, 11, 13, 15\}$.
    \item Case 7: Case 6 plus  5 more options with payoff $\phi_3$ and strike prices $K=\{5, 6, 7, 11, 13\}$. 
    \item\label{E19} Case 8: Case 7 plus 2 more options with payoff $\phi_4$ and strike prices $K=\{5, 7\}$.  
\end{enumerate} 

We can clearly observe now that the addition of more information, \textit{i.e.} more traded options with the same payoff but different strikes, is making the overall quality of the model-free bounds better.
This is particularly pronounced for Cases 6 and 8, which are now much closer to the reference value for large strikes (greater than 8) compared to Cases 2 and 4.
Another interesting observation is that the additional market information in Cases 3 and 7 does not provide a significant reduction of the model-free upper bound compared to the already existing information in Cases 2 and 6.
A possible explanation for this phenomenon could lie in the volatility of the marginal distributions; Assets 2 and 3 have significantly higher volatility than Asset 1, and probably the information provided in Cases 2 and 6 about their prices is more ``valuable''.
Finally, let us point that although the number of three-asset options in Cases 4 and 8 is exactly the same, the additional information on the lower-dimensional options is leading to a significant improvement of this bound, which is now almost touching the reference value for all strikes.

\subsection{Experiment 2}

Let us now consider a financial market that contains six traded assets $S=(S^1,\dots,S^6)$ and assume we are interested in the price of a European basket call option on these assets, \textit{i.e.} the payoff equals
\begin{align}
    f(x) = \bigg( \frac16 \sum_{i=1}^6 x_i - K \bigg)^+.    
\end{align}
The interest rate is assumed zero for simplicity, the maturity $T=1.5$ (years), while the initial values, variances and correlation matrix are summarized below
\begin{align*}
S_0 = \begin{pmatrix}
    10 \\
    10 \\
    10 \\
    10 \\
    10 \\
    10
    \end{pmatrix}, 
\quad
\sigma = \begin{pmatrix}
    0.3 \\
    0.4 \\
    0.5 \\
    0.35 \\ 
    0.45 \\
    0.55
    \end{pmatrix}, 
\quad    
\rho = \begin{pmatrix}
    1   & 0.45&  0.35&  0.44&  0.5&  0.30 \\ 
    0.45&  1&  0.38&  0.36&  0.41&  0.43\\ 
    0.35&  0.38&  1&  0.44&  0.32&  0.42 \\
    0.44&  0.36&  0.44&  1&  0.46&  0.29\\
    0.5 &  0.41&  0.32&  0.46&  1&  0.6\\
    0.30&  0.43&  0.42&  0.29&  0.6 & 1\\  
\end{pmatrix}
\end{align*}





\subsubsection*{Impact of additional information}

In order to understand the impact that additional information has on the distance between the model-free upper bound and the reference price, we consider the following setting, where information is added gradually, \textit{i.e.}
\begin{enumerate}[label=(E2.\arabic*),itemindent=1.em]\setcounter{enumi}{-1}
    \item\label{E20} Base case: only the six marginal distributions are known.
    \item Case 1: Base plus call-on-min options with payoff $(x_1 \wedge \dots \wedge x_6-K)^+$ for 8 strike prices
        $$K=\{6.5, 7.5, 8.5, 9.5, 10.5, 11.5, 12.5, 13.5\}.$$
    \item Case 2: Case 1 plus call-on-max options with payoff $(x_1 \vee \dots \vee x_6-K)^+$ for 8 strike prices
        $$K=\{6.5, 7.5, 8.5, 9.5, 10.5, 11.5, 12.5, 13.5\}.$$
    \item\label{E23} Case 3: Case 2 plus basket options with payoff 
        \[  
            \Big( \frac15 \sum_i x_i - K \Big)^+ \text{ for } i \in \{1,\dots,5\}, \ i \in \{2,\dots,6\}, \ i \in \{1,2,3,5,6\},
        \]
        each with 8 strike prices $K=\{6.6, 7.6, 8.6, 9.6, 10.6, 11.6, 12.6, 13.6\}.$   
\end{enumerate}
The outcome of this experiment is presented in Figure \ref{fig:ex2_1}.
We can observe that the addition of more information in the form of traded asset prices is making the model-free bounds sharper, and gradually directs them closer to the reference price.
We can also observe though that the addition of basket option prices in Case 3 is having a significant impact, in particular for large strikes, greater than ten.

\begin{figure}[H]
    \centering
    \includegraphics[width=0.6\textwidth]{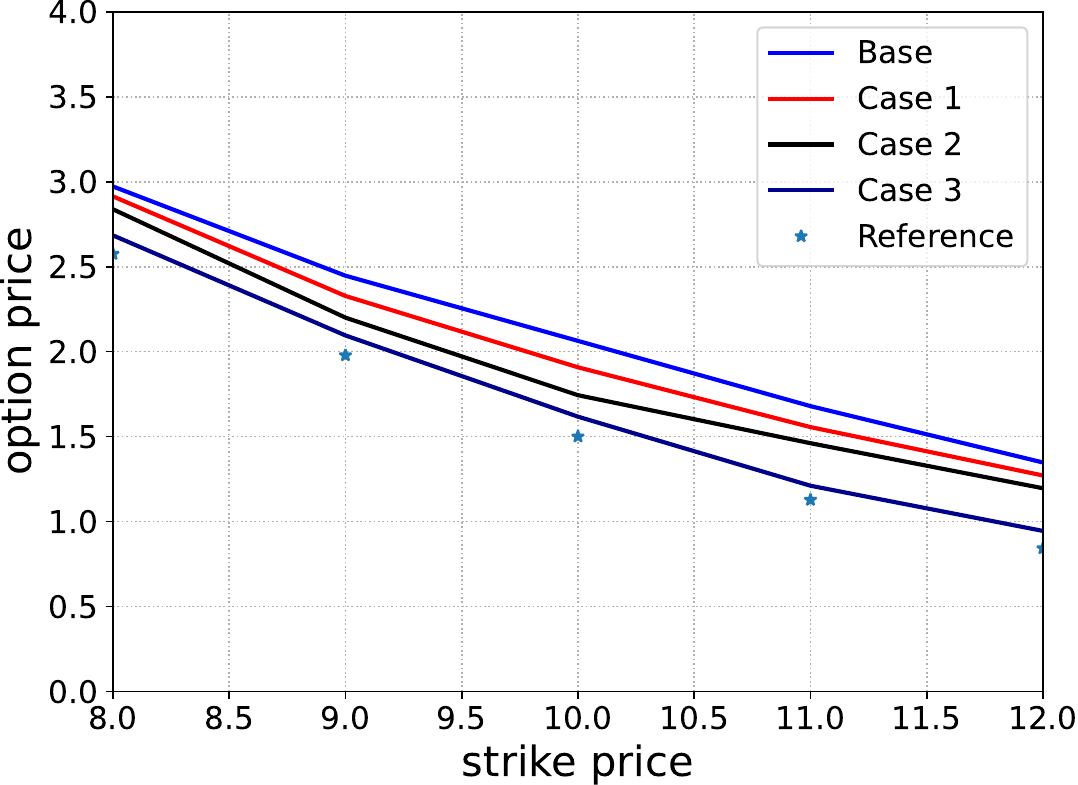}
    \caption{Model-free bounds for various strikes using the setting \ref{E20}--\ref{E23}.} 
    \label{fig:ex2_1}
\end{figure}


\subsubsection*{Impact of relevant information}

In order to investigate further the last observation, that the addition of `relevant' information has a higher impact than other information, whereby with `relevant' we refer to options that have the same payoff structure as the target payoff function $f$, we have conducted the following experiment.
Here, information is added gradually in the first four steps, while in the last step \textit{only} the relevant information is taken into account together with the marginals.
In other words, the setting is the following:
\begin{enumerate}[label=(E2.\arabic*),itemindent=1.em]\setcounter{enumi}{-1}
    \item Base case: only the six marginal distributions are known.\setcounter{enumi}{3}
    \item\label{E24} Case 4: Base plus three put-on-min options with payoff $(K - x_i  \wedge x_j)^+$ for $\{i,j\} \in\big\{ \{1,2\}, \{3,4\}, \{5,6\} \big\}$, each  with  8 strike prices
        $$K = \{ 6.75, 7.75, 8.75, 9.75, 10.75, 11.75, 12.75, 13.75\}.$$
    \item Case 5: Case 4 plus call-on-min options with payoff $(x_1 \wedge \dots \wedge x_6 - K)^+$ for 8 strike prices
        $$K = \{ 6.5, 7.5, 8.5, 9.5, 10.5, 11.5, 12.5, 13.5 \}.$$
    \item Case 6: Case 5 plus call-on-max options with payoff $(x_1 \vee \dots \vee x_6-K)^+$ for 8 strike prices
        $$K = \{ 6.5, 7.5, 8.5, 9.5, 10.5, 11.5, 12.5, 13.5 \}.$$
    \item Case 7: Case 6 plus basket options with payoff 
        \[  
            \Big( \frac15 \sum_i x_i - K \Big)^+ \text{ for } i \in \{1,\dots,5\}, \ i \in \{2,\dots,6\}, \ i \in \{1,2,3,5,6\},
        \]
        each with 8 strike prices $K=\{6.6, 7.6, 8.6, 9.6, 10.6, 11.6, 12.6, 13.6\}.$
    \item\label{E28} Case 8: Base plus basket options with payoff 
        \[  
            \Big( \frac15 \sum_i x_i - K \Big)^+ \text{ for } i \in \{1,\dots,5\}, \ i \in \{2,\dots,6\}, \ i \in \{1,2,3,5,6\},
        \]
        each with 8 strike prices $K=\{6.6, 7.6, 8.6, 9.6, 10.6, 11.6, 12.6, 13.6\}.$
\end{enumerate}

The outcome of this experiment is presented in Figure \ref{fig:ex2_2}.
We can observe once again that the addition of more information is improving the quality of the model-free bounds, in the sense that they are getting closer to the reference price.
The most interesting observation however is the comparison between Case 7 and 8: as can be see from \cref{fig:ex2_2}, the line corresponding to these two cases are almost touching, although Case 7 contains significantly more information than Case 8.
This leads us to the conclusion that relevant information, in the sense of option prices with the same payoff function, provide significant improvement of the model-free bounds.
Therefore, given the perennial trade-off between accuracy and computational efficiency in mathematical finance, relevant information should be prioritized in modeling and applications, even at the expense of other information, since the loss in accuracy appears small.

\begin{figure}[H]
    \centering
    \includegraphics[width=0.6\textwidth]{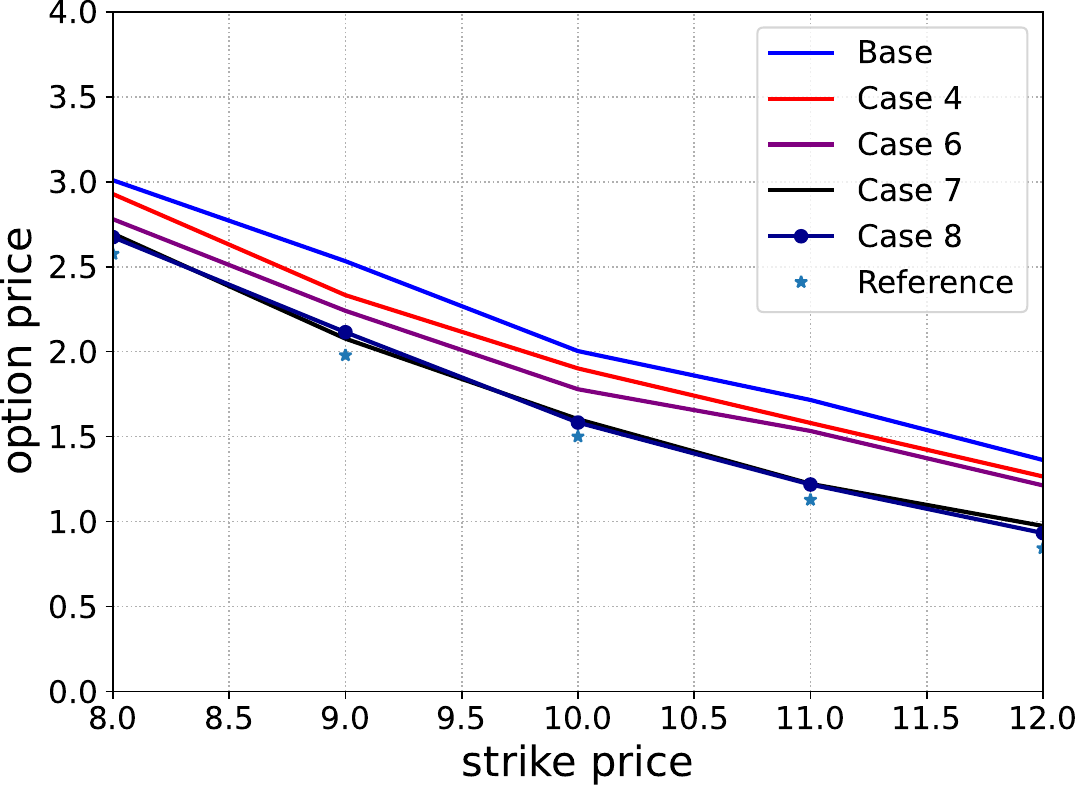}
    \caption{Model-free bounds for various strikes using the setting \ref{E20}, \ref{E24}--\ref{E28}.} 
    \label{fig:ex2_2}
\end{figure}


%


\subsection{Computational time}

Finally, we would like to study the computational times for the numerical approximation of the model-free bounds and, most importantly, understand how the computational times scale with the dimension of the problem, \textit{i.e.} the number of traded primary assets.
The results of this experiment are summarized in \cref{table:comptime}, and concern experiments conducted on a Windows-system laptop with an AMD Ryzen$^\text{TM}$ 7 5800HS processor (3201MhZ).

The experiments show that the method is quite fast and efficient, since computations for 18 assets can be conducted in approximately 18 minutes. 
More importantly, we can observe that the computational time grows linearly with the number of traded assets, which offers the hope to tackle problems in even higher dimensions with this methodology. 

\begin{table}[H] 
\begin{tabular}{cc} 
\# Assets & Time (sec)  \\
\hline\hline
6               & 498   \\
15              & 921   \\
18              & 1107  \\
\hline
\end{tabular} 
\caption{Dimension vs computational time.}
\label{table:comptime}
\end{table}


\bibliographystyle{abbrvnat} 
\bibliography{References} 


\end{document}